\documentclass{article}

\usepackage{PRIMEarxiv}

\usepackage[utf8]{inputenc} 
\usepackage[T1]{fontenc}    
\usepackage{hyperref}       
\usepackage{url}            
\usepackage{booktabs}       
\usepackage{amsfonts}       
\usepackage{nicefrac}       
\usepackage{microtype}      
\usepackage{lipsum}
\usepackage{fancyhdr}       
\usepackage{graphicx}       
\graphicspath{{media/}}     
\usepackage{amssymb,amsmath,amsthm}

\newtheorem{theorem}{Theorem}

\usepackage{algorithmic}
\usepackage[linesnumbered,algoruled,boxed,lined]{algorithm2e}
\newtheorem{definition}{Definition}
\newtheorem{example}{Example}
\newtheorem{proposition}{Proposition}
\newtheorem{lemma}{Lemma}
\newtheorem{corollary}{Corollary}
\newcommand{\T}{{\mathcal T}}

\usepackage{graphicx}
\usepackage{hyperref}

\title{Lockless Blockchain  Sharding with Multiversion Control

}

\author{
  Ramesh Adhikari \\
  School of Computer
and Cyber Sciences\\
  Augusta University \\
  Augusta, GA 30912, USA \\
  \texttt{radhikari@augusta.edu} \\
   \And
  Costas Busch \\
   School of Computer
and Cyber Sciences\\
  Augusta University \\
  Augusta, GA 30912, USA \\
  \texttt{kbusch@augusta.edu} \\
}

\begin{document}
\maketitle

\begin{abstract}
Sharding is used to address the performance and
scalability issues of the blockchain protocols, which divides the overall transaction processing costs among multiple clusters of nodes. Shards require less storage capacity and communication and computation cost per node than the existing whole
blockchain networks, and they operate in parallel to maximize performance.
However,
existing sharding solutions use locks for
transaction isolation which lowers the system
throughput and may introduce deadlocks. In this paper, we propose a lockless transaction method for
ensuring transaction isolation without using locks, which improves the concurrency and throughput of the transactions. In our method, transactions are split into subtransactions to enable parallel processing in multiple shards. We use versions for the transaction accounts to implement consistency among the shards. We provide formal proof for liveness and correctness. We also evaluate experimentally our proposed protocol and compare the execution time and throughput with lock-based approaches. The experiments show that the transaction execution time is considerably
shorter than the lock-based time and near to the ideal (no-lock) execution time. 
\end{abstract}


\section{Introduction}
The popularity of blockchains has grown due to their numerous benefits in decentralized applications. They have several special features such as  fault tolerance, transparency, non-repudiation,  and immutability \cite{survey-of-onsensus}. To maximize bandwidth usage, every transaction is hashed with a cryptographic function and multiple transactions are divided into blocks \cite{packing-message-as-a-tool}.
After that, a ledger is created by chaining all the blocks together using a consensus mechanism to append blocks.
Assuming that nobody else can be trusted, every node is in charge of keeping its own copy of the distributed ledger.
As a result, if someone or some system attempts to alter or restore a portion of these transactions it will be detected, which provides assurances of data integrity and finality.

The distributed cryptocurrency blockchain system known as Bitcoin \cite{bitcoin} is one of the first and most well-known instances of how blockchain was originally designed for the reliable exchange of digital goods. 
A permissionless blockchain allows anyone to join or leave the network without having to reveal their true identity.
No participant can be truly trusted in such situations.
Due to the lack of identity, a computationally intensive consensus process called proof-of-work that is based on cryptography is required.
On the other hand, in permissioned systems
the environment is more controlled and 
allows for more power-efficient consensus  protocols
based on Byzantine agreement \cite{SharPer}; nevertheless, even these blockchain protocols do not scale well due to large communication overhead.

Unfortunately, conventional blockchain applications have a fully replicated architecture where each node stores a copy of the whole blockchain and processes every transaction which causes scalability issues in contemporary very big data-based applications \cite{blockchain-hype-vs-reality}. When the number of transactions and storage nodes increases, the blockchain system not only takes a longer time to achieve a consensus among nodes but also takes more time to process the transaction; therefore, it reduces the overall performance of the system. To mitigate the scalability issue of the blockchain, several blockchain protocols like  Elastico \cite{Elastico}, OmniLedger \cite{OmniLedger},  RapidChain \cite{Rapidchain}, and ByShard \cite{Byshard} has proposed to introduce sharding to provide scalability which divides the whole replicated single blockchain system to multiple shards and each shard processes its own transactions independently.

The blockchain nodes are divided into clusters of nodes called shards.
Subsets in each shard may contain Byzantine nodes.
We presume that each shard employs a BFT (Byzantine Fault Tolerant) consensus algorithm with authentication, such as PBFT \cite{PBFT}.
Existing sharding solutions such as, Elastico \cite{Elastico}, OmniLedger \cite{OmniLedger}, and RapidChain \cite{Rapidchain} are tailored for supporting open (or permissionless) cryptocurrency applications and are not easily generalizable to other systems. 
To address system-specific specialized approaches towards sharding, {Hellings {\textit{et al.}}} \cite{Byshard}, introduced ByShard, and combine two conventional sharded database concepts, two-phase commit and two-phase locking for  atomicity and isolation of transactions in a Byzantine environment. However, their sharding solutions are not optimal as the locks are expensive, and when a process locks a data set for reading or writing, all other processes attempting to access the same data set are blocked until the lock is released, which lowers system throughput.

In this paper, we propose a different method for ensuring transaction isolation without using locks, which improves the transaction processing time. We propose a novel algorithm for ensuring atomicity and isolation of the transactions in the distributed environment. 

\subsection{Contributions}
To our knowledge, we give the first lockless approach to blockchain sharding.
We provide the following contributions:
   \begin{itemize}
     \item 
     We provide a lockless protocol for sharded blockchains. Our protocol is based on multi-version concurrency control of the various shared objects (accounts) that the transactions access. 
     A transaction is first split into subtransactions that execute in parallel in multiple shards. Using object versioning, the subtransactions can detect whether there is a conflict with other subtransactions that attempt to access the same shared objects concurrently. In case of a conflict, a transaction may need to restart and attempt to commit again. 
     
     \item We provide correctness proofs for the safety and liveness of our 
     proposed protocol.
     We also evaluate our protocol experimentally through simulations and we observe that the transaction execution time is considerably faster than the lock-based approaches and also the throughput of the transactions is improved with an increasing number of shards.
   \end{itemize}

\paragraph{\bf Paper Organization:}
The rest of this paper is structured as follows: Section \ref{sec:related-work} provides previous related works. Section \ref{sec:preliminaries-and-sharding-model} describes the preliminaries for this study and the sharding model. Section \ref{sec:proposed-solution} discusses our proposed lockless sharding protocol. In \ref{section:formal} we provide the correctness analysis. Section \ref{sec:performance-evaluation} discusses the performance evaluation of our work, experimental setup, and experimental results. Finally, we give our conclusions in Section \ref{sec:conclusion}.

\section{Related Work}
\label{sec:related-work}
Several proposals have come forward to address the blockchain scalability issue in the consensus layer \cite{Jalal-Window,jalalzai2021hermes,No-Commit-Proofs,jalalzai2019proteus,SBFT}. Although these protocols have addressed the scalability issues to some extent, the system still cannot maintain good performance as the network size grows too large (thousand or more node participants). The sharding technique has been used to further improve the scalability of a blockchain network.
Sharding is a fundamental concept in databases
which has been recently used to improve the efficiency of blockchains \cite{Byshard,Rapidchain}.

 The way that conventional big database systems achieve scalability is by separating the whole database into shards (or partitions) \cite{distributed-database-sharding}, which increases the efficiency of the system by dividing the workload among the shards. To ensure ACID characteristics \cite{transaction-processing} of the database transactions, coordination is needed among the multiple shards if the transaction access multiple shards objects. In the distributed database system two-phase commit (2PC), and  two-phase locking (2PL) \cite{two-phase-locking}, are used for atomicity, concurrency control, and isolation for the transactions. And we achieved these characteristics in our model in a different way without using locks.

Similarly, the blockchain network can be split up into smaller committees using the well-researched and tested technique of sharding, which also serves to scale up databases and lower the overhead of consensus algorithms. Elastico \cite{Elastico}, OmniLedger \cite{OmniLedger}, and RapidChain \cite{Rapidchain} are a few examples of sharded blockchains. These methods are not generalizable to other applications since they concentrate on a simple data model, that is  the unspent transaction output (UTXO) model \cite{Cerberus}. In addition, these methods use locks for the isolation of transactions.
As in databases, a blockchain transaction must be isolated since it interacts with the global state. In reality, it is necessary to avoid dirty, phantom, or unrepeatable reads \cite{sql-isolation-level}. Additionally, transactions must comply with all of the ACID properties \cite{transaction-processing}.
 Typically, two-phase locking \cite{two-phase-locking} is used to accomplish optimistic concurrency control \cite{optimistic-menhods-for-concurrency-control}, serializable snapshot isolation \cite{serializable-isolation,serializable-snapshot-isolation}. 

To mitigate the system-specific specialized approaches towards sharding, {Hellings {\textit{et al.}}} \cite{Byshard}, propose ByShard. It uses a two-phase commit to ensure the atomicity of the transaction and two-phase locking for isolation of the transaction in a Byzantine environment of the blockchain system. However, locking is expensive because when a process locks a data set for reading or writing, all other processes attempting to access the same data-set are blocked until the lock is released, which lowers system throughput. An innovative lock-free method for ensuring transaction isolation is presented by {Hagar Meir {\textit{et al.}}} \cite{lockless-transaction-hyperledger-fabric}, In order to construct version-based snapshot isolation, it takes advantage of the key-value pair versioning that already exists in the database and is mostly utilized at the validation phase of the transaction to detect the read-write \cite{lockless-transaction-hyperledger-fabric}. However, they are not addressing their solution in a sharding-based blockchain model.

In our solution, we are using a lockless approach to achieve transaction isolation with sharding. We use multiversion concurrency control~\cite{multiversion-concurrency-control},  as we describe in our proposed model in Section \ref{sec:proposed-solution}.

\section{Preliminaries and Sharding Model}
\label{sec:preliminaries-and-sharding-model}
\paragraph{\bf Blocks and Blockchain:}
Similar to a page in a record-keeping book (ledger), a block in a blockchain is a data structure that contains a sequence of transactions. A block has a header with additional metadata, including the block hash and a Merkle root hash for all the transactions in the block.
Blockchain is a chain (linked list) of blocks, which is an immutable, distributed, decentralized, peer-to-peer ledger replicated over several nodes connected in a network that allows the recording of transactions.

\paragraph{\bf Shards:}
We assume that the system consists of a set of $N$ (replica) nodes, where $n = |N|$. We design a sharded system as a partitioning of the $N$ nodes into $w$ shards $S_1, S_2,\dots, S_w$, where $N = \cup_i S_i$ and each $S_i \subseteq N$ is a subset of the nodes such that $S_i \cap S_j = \emptyset$, for $i \neq j$. Let $n_i = |S_i|$ represent the number of replica nodes in shard $S_i$, and $f_i$ represent the number of Byzantine nodes in shard $S_i$. Similar to related work \cite{Rapidchain}, to achieve Byzantine fault tolerance in each shard we assume that $n_i > 3 f_i$ within each shard. Hence, we focus on the consistency aspects, assuming there is an underlying consensus protocol in each shard. 

Let $\mathcal{O}$ be a set of shared objects that are accessed by the transactions.
Similar to related work \cite{Byshard}
we assume that every shard is responsible for a subset of the shared objects (accounts) that are accessed by the transactions.
Namely, $\mathcal{O}$ is partitioned into 
subsets $\mathcal{O}_1, \ldots, \mathcal{O}_w$,
where $\mathcal{O}_i$ is the set of objects handled by shard $S_z$.
Every shard $S_i$ maintains its own local ledger (local chain) $L_{i}$ and runs a local consensus algorithm to achieve this (e.g. PBFT \cite{PBFT}). 
The shard $S_i$ processes subtransactions  related to the object set $\mathcal{O}_{i}$ (see below for the description of subtransactions).
The local chains define implicitly the global blockchain, that is,
the global order of all transactions is implied by the order of their respective subtransactions in the local chains. 

\paragraph{\bf Timing Assumptions:}
We consider a semi-synchronous setting where communication delay is upper bounded by some time $\Delta_1$,
which means that every message is guaranteed to be delivered within $\Delta_1$ time.
Our sharding protocol does not require knowledge of $\Delta_1$.
We assume that every transaction has a unique ID based on its generation timestamp,
hence IDs grow over time.
Due to the semi-synchronous model,
since local clocks are not perfectly synchronized,
we assume a new timestamp (generated at any node in $N$) 
will be strictly larger than any timestamp generated $c \cdot \Delta_1$ 
time earlier (where the constant $c$ depends on the system).
Hence, we assume that for a transaction $T_i$ that arrives
at time $t$, any other transaction $T_j$ that arrives after $t + c \cdot \Delta_1$ 
will have always a higher ID than $T_i$.

For guaranteeing liveness in our protocol,
we assume that each $\Delta_2$ time each shard sends the lowest transaction ID 
from its transaction pool to other shards.
Here $\Delta_2$ is known to each shard (in order to periodically perform the lowest ID transmission) but is not related to $\Delta_1$.
In this way, each shard maintains the set of lowest transaction IDs which are periodically updated
with new lowest ID information from each shard. 
The transaction which has the global lowest ID
gets within a bounded time high priority and is eventually added to the blockchain.
The process of propagating the lowest IDs is running in the background 
while the normal execution phases take place.


Similar to previous works \cite{Byshard}, 
we assume that each shard runs locally a PBFT-style \cite{PBFT} consensus algorithm 
in every phase of our algorithm
which takes bounded time $\Delta_3$ for decisions (e.g. in \cite{Byshard} it is assumed that $\Delta_3 = 30ms$).
Our protocol does not need to know $\Delta_3$.

\paragraph{\bf Subtransactions:}

We model each transaction $T_i$ which consist of subtransactions $T_{i,k_1},\ldots,T_{i,k_j}$, such that:
   \begin{itemize}
     \item Subtransaction $T_{i,k_l}$ uses only objects in $\mathcal{O}_{k_l}$ in shard $S_{k_l}$. We also say that the subtransaction $T_{i, k_l}$ belongs to shard $S_{k_l}$.
     \item The subtransactions of a transaction $T_i$ do not depend on each other and can be executed in parallel in any relative order.
     \item A subtransaction consists of two parts: (i) condition checking, where various explicit conditions on the objects are checked, and (ii) updates on the objects. 
   \end{itemize}

\begin{example}
Consider a transaction ($T_1$) consisting of read-write operations on the accounts with several conditions.

$T_1$ = ``Transfer 2000 from Rock account to Asma account, if Rock has 3000 and Asma has 500 and Mark has 200''. We split this transaction into three subtransactions, where shards $r$, $a$, and $m$ are responsible for the respective accounts of Rock, Asma, and Mark:
\end{example}

\begin{minipage}[t]{0.415\columnwidth}
\begin{itemize}
\item[$T_{1,r}$]:
``Check Rock has 3000''\\
: ``Remove 2000 from Rock account''\\
\item[$T_{1,m}$]: ``Check Mark has 200''\\
\end{itemize}
\end{minipage}\hfill 
\begin{minipage}[t]{0.415\columnwidth}
\begin{itemize}
\item[$T_{1,a}$]:
``Check Asma has 500''\\
~~~: ``Add 2000 to Asma account''\\
\end{itemize}
\end{minipage}

After splitting the transaction $T_1$ into its subtransactions we send each subtransaction to its respective shard associated with that account. If the conditions are satisfied (for example in $T_{1,r}$ check if Rock has 3000) and the transaction is valid (for example in $T_{1,r}$ Rock has indeed 2000 in the account to be removed) then the destination shards are ready to commit the subtransactions which imply that $T_1$ will commit as well. Otherwise, if any of the conditions in the subtransactions is not satisfied or the subtransactions are invalid, then 
the corresponding subtransactions abort, which results in $T_1$ aborting as well. In this case, all subtransactions of $T_1$ will also be forced to abort. 

\section{Sharding Algorithm}
\label{sec:proposed-solution}

Our sharding protocol consists of two parts, the leader shard algorithm (Algorithm \ref{alg:one}), and the destination shard algorithm (Algorithm \ref{alg:two}).
 
Every transaction has a designated {\em leader shard}, which will handle its processing.
Each leader shard has a transaction pool for all the transactions that have it as their leader.
The job of the leader shard is to pick a transaction from the transaction pool and split it into subtransactions and send them to {\em destination shards}. The leader shard interacts with the destination shards through a protocol with seven phases which decide whether the subtransactions they receive are able to commit locally or not. 
The leader shard picks the transaction from its transaction pool on the basis of the priority of the transactions so that the earliest transaction (i.e. with lower ID) proceeds first.
Whereas the destination shard checks each received subtransaction and if it is valid then it commits it and appends it to its local ledger, otherwise, it aborts the subtransaction and sends the corresponding message to its leader shard.

To achieve transaction isolation, we use multi-version concurrency control~\cite{multiversion-concurrency-control} in each destination shard, which saves multiple versions of each object (account) so that data can be safely read and modified simultaneously. 
When a destination shard processes a subtransaction, it takes a snapshot of the current version of each object that the subtransaction will access. 
When the subtransaction is about to commit, it compares the latest version with the recorded snapshot version. If these are the same then the subtransaction is eligible to commit; otherwise,
the subtransaction cannot commit. 
The leader shard is informed accordingly from the destination shards. If all subtransactions are eligible to commit then the whole transaction will commit and is removed from the leader shard pool. If however, a subtransaction is not eligible to commit, the whole transaction will restart and is reinserted back into its pool.

In our algorithm, each transaction whose conditional statements are satisfied will eventually commit (as we show in the correctness proofs).
Our algorithm may attempt to commit the transaction multiple times by restarting it in case of conflicts with other transactions. However, if the condition of a transaction is not satisfied then the transaction is aborted and will not restart (is removed completely from the pool). Using the object versions the algorithm guarantees safety, as it does not allow conflicting transactions to commit concurrently. To ensure liveness, the algorithm prioritizes earlier generated transactions.
\subsection{Detailed Algorithm}
We now describe the details of our protocol in Algorithms \ref{alg:one} and \ref{alg:two}.
Our combined protocol consists of seven phases. 
As mentioned in Section \ref{sec:preliminaries-and-sharding-model},
to ensure liveness, periodically every $\Delta_2$ time, each leader shard sends the lowest ID from its transaction pool to every other shard. So that in case of conflict, 
priority is given to the transaction with the smallest known ID. 
In this way, each destination shard maintains in $T''_l$ (Algorithm \ref{alg:two})
the lowest known ID that it received from all leader shards.
If a subtransaction realizes that it belongs to a transaction with the smallest ID in the system then 
it gets the highest priority and enforces itself to commit.
This is further achieved with the help of a rollback mechanism that we discuss below.

Now we describe each phase of our algorithm.
For simplicity of presentation, 
we assume that each subtransaction accesses a single object in each destination shard.
However, our algorithms can be generalized for the case where each subtransaction accesses 
multiple objects.
  

\begin{algorithm}
\caption{Leader Shard $S_k$}
\label{alg:one}
Let $P_k$ be the pool of pending transactions in shard $S_k$\;
\tcp{Periodically, every $\Delta_2$ time the transaction with lowest ID in $P_k$ is sent to every other shard}
Let $C_k$ be the committed transaction pool;
\BlankLine
\tcp{Phase 1}
Pick transaction $T_i$ with lowest ID from $P_k$ and remove it from $P_k$\;

Split $T_i$ into subtransactions\;

Let $S(T_i)$ be the set of destination shards for the subtransactions of $T_i$\; 

Send each subtransaction $T_{i,j}$ to the corresponding destination shard $S_j$ (in parallel for all subtransactions of $T_i$)\;
\BlankLine

\tcp{Phase 3}
\uIf{{\em ``commit vote''} message is received from all shards in $S(T_i)$}
    {
    Send ``commit'' message to all shards in $S(T_i)$\;
    }\
\ElseIf{{\em ``abort vote''} message is received from any shard in $S(T_i)$}{
Send ``abort'' message to all shards in $S(T_i)$\;
}
\BlankLine
\tcp{Phase 5}
\uIf{{\em ``committed''} message is received from all shards in $S(T_i)$}{%
    Send ``release'' to all shards in $S(T_i)$\;
}\uElseIf{{\em ``restart vote''} message is received from any shard in $S(T_i)$}{
Send ``restart'' message  to all shards in $S(T_i)$\;
}\ElseIf{{\em ``aborted''} message is received from all shards in $S(T_i)$}{
  Transaction $T_i$ is discarded\;
}
\BlankLine
\tcp{Phase 7}
\uIf{{\em ``released''} message is received from all shards in $S(T_i)$}{
    Transaction $T_i$ has completed\;
    Add $T_i$ to $C_k$;
}\ElseIf{{\em ``restarted''} message is received from all shards in $S(T_i)$}{
Transaction $T_i$ is reinserted into the pool $P_k$ to be processed again\;
}

\BlankLine
\BlankLine
\tcp{Handling Force Rollback Messages}
\uIf{{\em ``force rollback $T'_{x,j}$''} message is received and $S_k$ 
is the leader shard of the respective transaction $T'$}{
  \uIf{$T' \in C_k$} {Get subtransaction information from $C_k$\;}
  Send respective ``force rollback $T'_{x,z}$'' to all destination shards of $T'$\; 
  }
  \uIf{{\em ``rollbacked $T'_{x,z}$''} message is received from all destination shards of $T'$ and $S_k$ is the leader shard of the transaction $T'$}{
    Insert $T'$ back in the pool $P_k$\;
    \If{$T' \in C_k$}{
        Remove $T'$ from $C_k$;
    }
  }
\end{algorithm}

\begin{algorithm}
\caption{Destination Shard $S_j$}
\label{alg:two}
$T''_l\leftarrow$ the lowest transaction ID from the IDs propagation process\;

\BlankLine
\tcp{Phase 2}
Subtransaction $T_{i,j}$ from leader shard $S_k$ is received\;
Suppose $T_{i,j}$ accesses object $O_d$\;
Let $R(O_d)$ and $W(O_d)$ be a set of transactions that
will respectively read or write $O_d$\;
Let $V_d$ be the latest version of object $O_d$\;
  \eIf{constraint match}{
Add $T_{i,j}$ to $R(O_d)$\;
\uIf{$T_{i,j}$ will write to $O_d$}{Add $T_{i,j}$ to $W(O_d)$\;}
 Send ``commit vote'' message to $S_k$\;
  }{
    Send ``abort vote'' message to $S_k$\;
  }
\BlankLine
\tcp{Phase 4}
    \uIf{{\em ``commit''} message is received from $S_k$}{%
        \uIf{ $(((W(O_d) \setminus \{T_{i,j}\} = \emptyset)$ or $(T_{i,j} \in W(O_d)$ and $R(O_d) \setminus \{T_{i,j}\} = \emptyset))$ and $($the latest version of object $O_d$ is still $V_d))$ or $(T_{i,j} = T''_l)$}
            {Append transaction $T_{i,j}$ to local chain $L_{j}$\;
             Send ``committed'' message to $S_k$\;
        
        \If{$T_{i,j} = T''_l$ \tcp{$T_{i,j}$ has the lowest ID in the system}}
        {For each $T'_{x,j} \in W(O_d)$ send ``force rollback'' message to its respective leader shard\;}
        }
        \Else{
         send ``restart vote'' message to $S_k$\;
        }
}\ElseIf{{\em ``abort''} message is received from $S_k$}{
  Send ``aborted'' message to $S_k$\;
}
\BlankLine
\tcp{Phase 6}

\uIf{ {\em ``restart''} message is received from $S_k$}{
Remove $T_{i,j}$ from $R(O_d)$ and $W(O_d)$ and from local chain $L_j$\;
Send ``restarted'' message to the leader $S_k$\;
}
\ElseIf{{\em ``release''} message is received from $S_k$}{
    \If{$T_{i,j}$ in $W(O_d)$} {
       Create new version $V_d+1$ for the object $O_d$\;
    }
Remove $T_{i,j}$ from $R(O_d)$ and $W(O_d)$\;
Send {\em ``released''} message to $S_k$\;
}
\BlankLine
\BlankLine
\tcp{Handling Force Rollback Messages}
\If{{\em ``force rollback $T'_{x,j}$''} message is received}{
        Remove $T'_{x,j}$ from $R(O_d)$ and $W(O_d)$\;
  Let $Z$ be the suffix in local chain $L_j$ starting from $T'_{x,j}$\; 
  Remove from $L_j$ the suffix $Z$ and send ``rollbacked $T'_{x,j}$'' message to its leader shard\;
  For each subtransaction $T'_{x,j}$ in $Z$, send ``force rollback $T'_{x,j}$'' message to the leader shard of $T'_{x,j}$\; 
  }
\end{algorithm}

    


{\bf Phase 1:} (Algorithm \ref{alg:one}) the leader shard ($S_k$) picks a transaction with the lowest transaction ID from its transaction pool ($P_k$) and  splits that transaction $T_i$ into its subtransaction $T_{i,j}$ and sends each $T_{i,j}$ to corresponding destination shards ($S_j$) in parallel.

{\bf Phase 2:} (Algorithm \ref{alg:two}) after receiving the subtransaction $T_{i,j}$ in destination shard ($S_j$) accessing an object, say $O_d$, it takes the latest version (say $V_d$) of the object $O_d$. After that, it checks the conditions (constraints) of the subtransaction $T_{i,j}$. If the constraints match (means subtransaction is eligible to commit) then, it adds the $T_{i,j}$ to the read set $R(O_d)$ and if $T_{i,j}$ will also write to $O_d$ then it also adds $T_{i,j}$ to write set $W(O_d)$ and sends a ``commit vote'' to the leader shard $S_k$. Otherwise, it sends a ``abort vote'' to the leader shard.

{\bf Phase 3:} (Algorithm \ref{alg:one}) the leader shard $S_k$ collects the votes from all the destination shards, and if it gets all ``commit vote'', (that means constraints are matched in all respective destination shards) then it sends the ``commit'' message to the corresponding destination shards. Similarly, if it gets any ``abort vote'' then it sends an ``abort'' message to all respective destination shards. 

{\bf Phase 4:} (Algorithm \ref{alg:two}) if the destination shard receives a ``commit'' message from a leader shard then, it checks the read set ($R(O_d)$) and write set ($W(O_d)$) of the accessing object and also checks the version of the object. If the subtransaction $T_{i,j}$ is only reading the object $O_d$ and the latest version of the object $O_d$ is still the same (i.e. $V_d$) then the shard appends this subtransaction to its local ledger $L_j$ and sends ``committed'' message to the leader. Similarly, if $T_{i,j}$ is trying to update object $O_d$ and the read set only contains $T_{i,j}$ (i.e. $(T_{i,j} \in W(O_d)$ and $R(O_d) \setminus \{T_{i,j}\} = \emptyset)$) and the latest version of the object $O_d$ is still same as the previously taken version (i.e. $V_d$) (that means the object is not modified by other transactions) then it does the necessary update operation and adds the subtransaction $T_{i,j}$ to its local chain and sends the ``committed'' message to the leader shard. Moreover, if the transaction ID of subtransaction $T_{i,j}$ is equal to the lowest known transaction ID ($T''_l$), that means the current subtransaction $T_{i,j}$ is the earliest subtransaction among all and it has a higher priority to execute. So it appends the subtransaction $T_{i,j}$ to its local chain and sends a ``committed'' message to its leader shard, and also sends ``force rollback'' to the leader of the subtransaction which is in the write set ($W(O_d)$). Otherwise, it sends a ``restart vote'', which means there is another higher-priority transaction accessing the object $O_d$ and not released yet. Similarly, if it receives the ``abort'' message then it sends an ``aborted'' message to its leader shard.

{\bf Phase 5:} (Algorithm \ref{alg:one}) if it receives a ``committed'' message from all destination shards (means that eligible to commit subtransactions are added to their local chain) then it sends a ``release'' message to the respective destination shards to release the subtransactions from their read set, write set and also to update the version of the object if required. Similarly, if the leader receives any ``restart vote'' message from any shards then it sends the ``restart'' message to the respective destination shards because some of the shards may have appended the subtransaction to their local chain so that should be removed, and restart should be consistent in all shards. Moreover, if it receives an ``aborted'' message from all destination shards, then transaction $T_i$ is discarded. 

{\bf Phase 6:} (Algorithm \ref{alg:two})  if the destination shard receives a ``restart'' message from the leader shard, then it removes the transaction $T_{i,j}$ from 
$R(O_d)$ and $W(O_d)$ and also removes $T_{i,j}$ from its local chain $L_k$ if it already added and sends ``restarted'' message to the leader. Similarly, if it receives a ``release'' message and $T_{i,j}$ is already in $W(O_d)$ that means it updated the object $O_d$ so it creates the new version of the object as $V_d +1$. After that, it removes $T_{i,j}$ from $R(O_d)$ and $W(O_d)$ and sends a ``released'' message to its leader shard.

 {\bf Phase 7:} (Algorithm \ref{alg:one}) if the leader shard receives a ``released'' message from all destination shards that means the transaction $T_i$ is completed, so it adds $T_i$ to the pool of committed transactions ($C_k$) so that it can get all the subtransaction information of $T_i$ in case of rollback. Similarly, if it receives a ``restarted''  message from all destination shards, then this transaction needs to be processed again, and is reinserted into the transaction pool $P_k$.

{\bf Handling Rollbacks:} This part of our protocol executes only in the special case (i.e. when the current transaction has the highest priority to execute than the already running transaction accessing the same object). After receiving the ``force rollback'' message from destination shards, the leader shard checks whether that subtransaction belongs to the committed transaction pool ($C_k$) or not. If the transaction of that subtransaction is in $C_k$  then it gets the other subtransaction information from $C_k$ otherwise it has the information about the currently running transaction, then it sends a ``force rollback'' message to all respective destination shards because it should be rollbacked in all the shards to be consistent. So if the destination shard receives the ``force rollback $T'_{x,j}$'' message from the leader then it rollbacks $T'_{x,j}$ from its shards and sends ``rollbacked $T'_{x,j}$'' message to its leader. Furthermore, if there exists some depending subtransaction $T'$ on $T'_{x,j}$ accessed the version of the object added by $T'_{x,j}$ then all depending transactions should be rollbacked and it sends ``rollback $T'$'' to its leader shard so this function executes recursively to rollback all the transactions which read the version of object added by $T'$. The leader shard collects the ``rollbacked'' messages from all the destination shards, and after receiving all the ``rollbacked'' messages from all the respective shards for the transaction $T'$, it adds $T'$ to its transaction pool to be processed again and removes $T'$ from the committed pool ($C_k$) if $T'$ is already in $C_k$.

\begin{example}
Consider two conflicting transactions $T_1$ and $T_2$ consisting of read-write operations on the accounts. We explain how our protocol handles these transactions.
\begin{itemize}
    \item[]$T_1$ = ``Transfer 2000 from Rock account to Asma account, if Rock has 3000 and Asma has 500 and Mark has 200''.\\
    \item[] $T_2$ = ``Transfer 500 from Asma account to Bob account, if Asma has 5000''.
\end{itemize}
\end{example}
Suppose leader shard ($S_{k_1}$) handles transaction $T_1$ and splits it into three subtransactions, where shards $S_r$, $S_a$, and $S_m$ are responsible for the respective accounts of Rock, Asma, and Mark. Similarly, the leader shard ($S_{k_2}$) handles transaction $T_2$ and splits it into two subtransactions where shard $S_a$ is responsible for the Asma account and shard $S_b$ is responsible for the Bob account.\\

\begin{minipage}[t]{0.415\columnwidth}
  \begin{itemize}
\item[$T_{1,r}$]: ``Check Rock has 3000''
\\: ``Remove 2000 from Rock account''\\
\item[$T_{1,a}$]:
``Check Asma has 500''\\
: ``Add 2000 to Asma account''\\
\item[$T_{1,m}$]: ``Check Mark has 200''\\

\end{itemize}
\end{minipage}\hfill 
\begin{minipage}[t]{0.415\columnwidth}
 \begin{itemize}

\item[$T_{2,a}$]: ``Check Asma has 5000''\\
: ``Remove 500 from Asma account''\\
\item[$T_{2,b}$]:
``Add 500 to Bob account''\\
\end{itemize}
\end{minipage}

Let us consider both leader shards ($S_{k_1}$ and $S_{k_2}$) are trying to execute their transaction in parallel. At this condition, there are two subtransactions accessing the same account of Asma (i.e. $T_{1, a}$, $T_{2, a}$) so there will be a conflict on the respective subtransactions. So, in our algorithm, each destination shard takes the version of every account and checks that version at the time of commit. 
In case of conflicts, the transaction that updated earlier the read and write sets ($R$ and $W$) at the destination shard will have a chance to commit.

\section{Correctness Analysis}
\label{section:formal}
Consider a set of transactions $\T = \{T_1, T_2, \ldots, T_\zeta\}$.
The objective is to arrange all transactions in $\T$
in a sequence $B = T_{i_1}, T_{i_2}, \ldots, T_{i_\zeta}$,
which is agreed upon by all non-faulty nodes in $N$.
We also write $T_{i_l} \prec_{B} T_{i_{l'}}$ for $l < l'$ to denote the relative order 
between two transactions in the sequence $B$.
The sharding system does not maintain the actual $B$ as a single blockchain (or ledger) explicitly,
but rather, the blockchain consists of a collection of local chains
which if combined they jointly give the whole blockchain $B$. 

Each shard $S_{\alpha}$ maintains a local chain
$L_{\alpha}$ of the sub-transactions $T_{i,\alpha}$
that it receives.
The subtransactions are appended in $L_{\alpha}$
according to the order that they commit in $S_{\alpha}$.
If $T_{i,\alpha} \prec_{L_{\alpha}} T_{j,\alpha}$,
and $T_{i,\alpha}$ conflicts with $T_{j,\alpha}$ (the two subtransactions conflict if they access the same object in $S_\alpha$ and one of the two is updating the object),
then we say that $T_{i,\alpha}$ {\em causes} $T_{j,\alpha}$
and we write $T_{i,\alpha} \to_{L_{\alpha}} T_{j,\alpha}$.

We define the {\em local chain system} as the tuple $L = (L_1, \ldots, L_{w})$ 
consisting of local chains in shards.
If $T_{i,\alpha} \to_{L_{\alpha}} T_{j,\alpha}$,
we can also simply write $T_{i,\alpha} \to_{L} T_{j,\alpha}$.
The casual relation $\to$ can be extended across two local chains $L_\alpha$ and $L_\beta$, $\alpha \neq \beta$, in the following way.
\begin{itemize}
\item
If $T_{i,\alpha} \to_{L_\alpha} T_{j,\alpha}$
and $T_{j}$ has a subtransaction $T_{j,\beta}$. 
\item
If $T_{i}$ has subtransactions $T_{i,\alpha}$ and $T_{i,\beta}$
such that $T_{i,\beta} \to_{L_\beta} T_{j,\beta}$.
\end{itemize}
In both cases, we say that
$T_{i,\alpha}$ causes $T_{j,\beta}$,
and we write $T_{i,\alpha} \to_L T_{j,\beta}$.
Consider from now on the transitive closure of the causal relation $\to_L$.

We say that the local chain system $L$ is {\em valid} 
if there is no subtransaction $T_{i,\alpha}$ 
such that $T_{i,\alpha} \to_L T_{i,\alpha}$.
That is, $L$ is a valid local chain system if there 
is no cyclic (transitive) causal relationship of a subtransaction to itself.

We say that a sequence $B$ is a {\em valid serialization} of the 
local chain system $L$ if $B$ is a sequence of all the subtransactions
which preserves the causal relationship of $L$.
Namely, if $T_{i,\alpha} \to_{L} T_{j,\beta}$ then $T_{i,\alpha} \prec_B  T_{j,\beta}$.
We say that a sequence $B$ is a {\em blockchain serialization} of $L$
if $B$ is a valid serialization of $L$, and for each transaction $T_i$
its subtransactions $T_{i,j_1}, \ldots, T_{i,j_{k}}$ appear consecutively in $B$ (without being interleaved by subtransactions of other transactions).
The goal is to show that our sharding protocol generates a local chain system $L$ that has a blockchain serialization $B$.
We introduce the {\em shard-coherence property} which we will use to prove the existence of $B$.

\begin{definition}[Shard-coherence]
We say that transactions $T_i$ and $T_j$ are shard-coherent with respect to local chain system $L$ if whenever
two of their subtransactions are casually related as $T_{i,\alpha} \to_L T_{j,\beta}$,
then for any two of their conflicting subtransactions $T_{i,\gamma}$ and $T_{j,\gamma}$
it holds that $T_{i,\gamma} \prec_{L_\gamma} T_{j,\gamma}$.
The local chain system $L$ is shard-coherent if every pair of transactions
are shard-coherent.
\end{definition}

The following result shows that in order to build a blockchain serialization $B$ from a chain system $L$,
it suffices to prove that $L$ is shard-coherent. 

\begin{proposition}
\label{prop:shard-coherent}
If a local chain system $L$ is shard-coherent, then $L$ has a blockchain serialization $B$.
\end{proposition}


The proof of Proposition \ref{prop:shard-coherent}
follows directly from Corollary \ref{corollary:valid-to-blockchain} and 
Lemma \ref{lemma:coherent-to-valid}
given below.
In the results below consider a local chain system $L = (L_1, \ldots, L_{w})$
for transactions $\T = \{T_1, T_2, \ldots, T_\zeta\}$.

\begin{lemma}
\label{lemma:valid-to-serializable}
If $L$ is a valid local chain system, then $L$ has a valid serialization.
\end{lemma}

\begin{proof}
Consider the sequence $A$ of subtransactions 
which is the concatenation of sequences 
$L_1, L_2, \ldots, L_w$.
Suppose that $A = a_1, a_2, \ldots, a_\delta$,
where $a_\sigma = T_{i_\sigma,j_\sigma}$
where $T_{i_\sigma,j_\sigma}$ is a subtransaction of transaction 
$T_{i_\sigma} \in \T$.

From $A$ we incrementally build a sequence $A'$ which is a valid serialization of $L$.
Let $A'_\sigma$ denote the sequence that we obtain 
after we appropriately insert (as explained below) the $\sigma$th element of $A$ into $A'$.
We prove by induction that $A'_\sigma$ is a valid serialization of 
the involved subtransactions of the respective induced subsystem 
$L^\sigma$ of $L$ 
that consists of the $\sigma$ subtransactions of $L$ under consideration
(the subsystem $L^\sigma$ keeps from each $L_i$ the involved subtransactions; note that the subsystem is valid).
The main claim follows when we consider $\sigma = \zeta$
which gives $A' = A'_{\sigma}$.

For the basis case $\sigma = 1$, and $A'_1 = a_1$
which is trivially a valid serialization of the single subtransaction.
Suppose that we built $A'_\sigma$ which is a valid serialization 
of the first $\sigma$ subtransactions in $A$, where $\sigma < \zeta$.

In order to build $A'_{\sigma+1}$ we take $a_{\sigma+1}$
and insert it into $A'_{\sigma}$, as follows.
Suppose that $A'_{\sigma} = a'_1, \ldots, a'_{\sigma}$.
If there is no $a'_i \in A'_{\sigma}$ such that $a_{\sigma+1} \to_L a'_i$,
then append $a_{\sigma+1}$ at the end of $A'_{\sigma}$, to obtain $A'_{\sigma + 1}$,
which is clearly a valid serialization.

Otherwise, 
let $a'_i$ be the earliest subtransaction in $A'_{\sigma}$ 
($i$ is the smallest index within $A'_{\sigma}$) 
such that $a_{\sigma+1} \to_L a'_i$,
and let $a'_j$ be the latest subtransaction in $A'_{\sigma}$ 
($j$ is the largest index within $A'_{\sigma}$) 
such that $a'_j \to_L a_{\sigma+1}$.
We examine two cases:
\begin{itemize}
\item
$j < i$: in this case we append $a_{\sigma+1}$ 
just before $a'_i$ (and clearly after $a'_j$) in $A'_{\sigma}$ to obtain $A'_{\sigma + 1}$,
which gives a valid serialization.
\item
$i < j$: we examine three sub-cases as follows.
\begin{itemize}
    \item $a'_i \to_L a'_j$: this case is impossible since this would create a cycle $a_{\sigma+1} \to_L a_{\sigma+1}$ in the causal relation $\to_L$, and hence, $L$ would not be valid, contradicting the assumptions.  
    
    \item $a'_j \to_L a'_i$: since $a'_i \prec_{A'_{\sigma}} a'_j$ this would imply that in $A'_{\sigma}$ is not a valid serialization of the involved subtransactions of $A$, which contradicts the induction hypothesis.
    
    \item {\em $a'_j$ and $a'_i$ are not related by $\to_L$ to one another}:
    consider the subsequence $s$ of $A'_{\sigma}$ from $a'_i$ to $a'_j$ (including $a'_i$ and $a'_j$). 
    Let $s_1$ be the subsequence of $s$ that includes all $a'_q$ such that $a'_i \to_L a'_q$;
    let $s_2$ be the subsequence of $s$ that includes all $a'_q$ such that $a'_q \to_L a'_j$;
    let $s_3$ be the remaining elements of $s$.
    Note that $s_1$ and $s_2$ are disjoint, 
    since otherwise $a'_i \to_L a'_j$.
    Next, we move all the elements in the sequence $s_2$ (keeping their relative order) to be before the first element in $s_1$. Moreover, 
    add $a_{\sigma+1}$ between the last element of $s_2$ and the first element of $s_1$.
    The resulting sequence $A'_{\sigma+1}$
    is clearly a valid serialization of the involved subtransactions.
\end{itemize}
\end{itemize}
\end{proof}

\begin{lemma}
\label{lemma:serializable-to-blockchain}
If the local chain system $L$ has a valid serialization, then $L$ is blockchain serializable.
\end{lemma}

\begin{proof}
Let $A$ be a valid serialization of $L$.
Suppose that $A = a_1, a_2, \ldots, a_\delta$,
where $a_\sigma = T_{i_\sigma,j_\sigma}$ 
and $T_{i_\sigma,j_\sigma}$ is a subtransaction of $T_{i_\sigma} \in \T$.

We will rearrange the subtransactions in $A$ to a new sequence $A'$
such that each transaction $T_i$ has its subtransactions 
consecutively in $A'$.
We will show how to do the transformation for a single transaction,
and this can repeat for the remaining transactions.

For a transaction $T_i$ let $T_{i, j_1}, \ldots, T_{i, j_q}$ denote its subtransactions,
with respective positions $a_{s_1}, \ldots, a_{s_q}$ in $A$.

From the validity of $A$ and transitivity of $\to_L$,
we have that if for some $l \in [q]$, $a_{j} \to_L a_{s_l}$, then $a_{j} \prec_A a_{s_{l'}}$, for every $l' \in [q]$.
Hence, if $a_{s_l}$ is the earliest subtransaction of $T_i$ in $A$ 
(i.e. $s_l$ has the smallest index among those with $l \in [q]$),
any $a_j$ that causes (through $\to_L$) any of the subtransactions of $T_i$
must appear in $A$ before $a_{s_l}$.
Therefore, we can move the subtransactions of $T_i$ 
and arrange them to appear consecutively starting 
at the position of $a_{s_l}$,
so that $a_{s_1}$ will take the place of $a_{s_l}$,
$a_{s_2}$ will appear immediately after $a_{s_1}$, 
and so on, until $a_{s_q}$.

Let $A'$ be the resulting sequence after we rearrange the subtransactions of $T_i$. 
Clearly, this transformation of $A$ has preserved its validity and also the 
subtransactions of $T_i$ appear consecutively in $A'$.
By repeating this process for each remaining transaction we obtain the final $A'$.
By induction (on the number of transactions),
it is clear that the final $A'$ is a blockchain serialization of $L$. 
\end{proof}

From Lemmas \ref{lemma:valid-to-serializable}
and \ref{lemma:serializable-to-blockchain}
we obtain the following corollary.

\begin{corollary}
\label{corollary:valid-to-blockchain}
A valid local chain system $L$ is blockchain serializable.
\end{corollary}

\begin{lemma}
\label{lemma:coherent-to-valid}
If a local chain system $L$ is shard-coherent, then $L$ is valid.
\end{lemma}

\begin{proof}
Suppose that $L$ is shard-coherent.
Suppose for the sake of contradiction that there is subtransaction $T_{i,j}$
such that $T_{i,j_k} \to_L T_{i,j_k}$
(that is, there is a cycle in $L$ with respect to causal relation $\to_L$).

Let $p = a_1, a_2, \ldots, a_\ell$ be a transitive ``relation path'', 
where each node in $a_i$ is a subtransaction of some transaction in $\T$
and $a_1 = a_\ell = T_{i,j_k}$,
and $a_i \to_L a_{i+1}$, for each $1 \leq i < \ell$.
Among all possible relation paths starting and ending to $T_{i,j_k}$,
let $p$ be the longest (and if there are multiple paths of the same longest length
then pick one of them arbitrarily).
Note that it has to be $\ell > 2$ 
since a subtransaction alone by itself cannot create 
cyclic dependencies.

First, consider the case where each subtransaction in $p$ 
is in the same shard $S_{\alpha} = S_{j_k}$ as that of $T_{i,j_k}$.
We consider two sub-cases:
\begin{itemize}
\item $a_{1} \prec_{L_{\alpha}} a_{2}$:
let $a_r$, $1 < r < \ell$, have the largest index $r$ such that $a_{1} \prec_{L_{\alpha}} a_{r}$.
Then clearly, $a_{r+1} \prec_{L_{\alpha}} a_{r}$ (note that $a_{r+1}$ exists since we took $r < \ell$ and also it holds $\ell > 2$).
However, since $a_{r} \to_L a_{r+1}$, the shard-coherence property of $L$ is violated between $a_r$ and $a_{r+1}$,
a contradiction.
\item $a_{2} \prec_{L_{\alpha}} a_{1}$:
since $a_{1} \to_L a_{2}$, the shard-coherence of $L$ is violated between $a_1$ and $a_2$,
a contradiction.
\end{itemize}

Next, consider the case where some subtransaction in $p$ is in
a different shard than $S_{\alpha}$.
Let $a_r$, where $1 \leq r < \ell$ 
be the first subtransaction (with the smallest index $r$) in $p$
which is in a different shard, say $S_{\beta}$, where $\alpha \neq \beta$.

We now show that $a_{r+1}$ must also be in $S_{\beta}$
conflicting with $a_r$.
Suppose to the contrary that $a_{r+1}$ is not in $S_{\beta}$.
Since $a_r \to_L a_{r+1}$,
there must be a subtransaction $T'$ in $S_{\beta}$
which conflicts with $a_r$,
such that $a_{r} \to_{L_\beta} T'$ and $T' \to_L a_{r+1}$.
However, this implies that $p$ can be augmented with $T'$,
which is a contradiction since $p$ is the longest relation path.
Thus, $a_{r+1}$ is in $S_{\beta}$.
Moreover, $a_{r+1}$ must be conflicting with $a_r$,
since otherwise we would find as above some other transaction $T'$ 
that conflicts with $a_{r}$ which could be inserted into $p$ to increase its length.
We examine two cases:
\begin{itemize}
    \item $a_r \prec_{L_{\beta}} a_{r+1}$: 
    from the cyclicity of path $p$ we have that $a_{r+1} \to_L a_{r}$ (going through $a_1$). 
    Hence, the shard-coherency is violated between $a_r$ and $a_{r+1}$,
    a contradiction.
    \item $a_{r+1} \prec_{L_{\beta}} a_{r}$:
    from $p$ we have that $a_{r} \to_L a_{r+1}$. 
    Hence, the shard-coherency is violated between $a_r$ and $a_{r+1}$,
    a contradiction.
\end{itemize}
\end{proof}


Next, we continue to show that in our sharding protocol two transactions that conflict
in the same shard, they cannot have some of their phases interleave.
\begin{lemma}
\label{lemma:interleave}
If two transactions $T_i$ and $T_j$ conflict in a destination shard $S_\gamma$,
and their respective subtransactions are processed concurrently by $S_\gamma$ 
so that they both go past phase 2 in $S_\gamma$ concurrently,
then at least one of the two transactions will restart or rollback.
\end{lemma}

\begin{proof}

Suppose transactions $T_i$ and $T_j$ have respective subtransactions $T_{i,\gamma}$ and $T_{j,\gamma}$ in $S_\gamma$. Moreover, suppose that these subtransactions conflict in $S_\gamma$ 
by accessing the same object $O_d$ and at least one of two of them is updating $O_d$.

Without loss of generality assume that $T_{i,\gamma}$ is updating $O_d$.
Suppose that $T_{i,\gamma}$ has finished executing phase $2$.
Hence, $T_{i,\gamma}$ has been added to write set $W(O_d)$. 
Then when $T_{j,\gamma}$ reaches phase 4, it will observe that $T_{i,\gamma}$
is already in $W(O_d)$ which will force $T_{j,\gamma}$ to restart.

On the other hand, if $T_{j,\gamma}$ has the lowest $ID$ then when 
$T_{j,\gamma}$ will reach phase $4$ it will force $T_{i,j}$ to rollback.
In either case, one of the two transactions will either restart or rollback.




\end{proof}

\begin{theorem}[Safety]
\label{theorem:safety}
The local chain system $L$ produced by our protocol has a blockchain serialization $B$.
\end{theorem}

\begin{proof}
From Proposition \ref{prop:shard-coherent},
we only need to prove that $L$ is shard-coherent.

Consider any two transactions $T_i$ and $T_j$ such that $T_{i,\alpha} \to_L T_{j,\beta}$.
Suppose that $T_{i,\gamma}$ and $T_{j,\gamma}$ conflict in shard $S_\gamma$
because they access at least one common object $O_{d}$ and one of the
two subtransactions updates $O_d$.
It suffices to show that $T_{i,\gamma} \prec_{L_\gamma} T_{j,\gamma}$.

Since $T_{i,\alpha} \to_L T_{j,\beta}$,
from the definition of the $\to_L$ relation,
we have 
that there is a sequence of transactions $T_{k_1}, T_{k_2}, \ldots, T_{k_z}$
with $T_{k_1} = T_i$, $T_{k_z} = T_j$
and $T_{k_i} \to T_{k_{i+1}}$, for $1 \leq i < z$,
such that any pair of consecutive transactions $T_{k_{l}}$ and $T_{k_{l+1}}$
have respective conflicting subtransactions
$T_{k_l, \delta}$ and $T_{k_{l+1}, \delta}$
on some common shard $S_\delta$ such that 
$T_{k_l, \delta} \prec_{L_\delta} T_{k_{l+1}, \delta}$.

Since $T_{k_l, \delta}$ and $T_{k_{l+1}, \delta}$
are appended in the local chain $L_{\delta}$,
while they both conflict, 
we have from Lemma \ref{lemma:interleave}
that they cannot go past phase 2 concurrently without one of them restarting or rolling back.
Therefore, $T_{k_l, \delta}$ finishes phase 6,
before $T_{k_{l+1}, \delta}$ enters phase 4.

This implies that phase 5 of $T_{k_l}$ (at its leader shard) finishes before 
phase 5 of $T_{k_{l+1}}$ starts (at its leader shard).
Therefore,
by induction, we can easily show that the end of phase 5 of $T_{i}$ (at its leader shard) 
occurs earlier than the beginning of phase 5 of $T_{j}$ (at its leader shard).



Suppose now that $T_{j,\gamma} \prec_{L_\gamma} T_{i,\gamma}$.
Since $T_i$ and $T_j$ commit in $S_\gamma$ and also conflict in $S_\gamma$ 
by sharing the same object,
then from Lemma \ref{lemma:interleave},
phase 6 of $T_j$ ends before phase $4$ of $T_i$ starts in $S_\gamma$.
Therefore, phase 5 of $T_j$ ends before phase 5 of $T_i$ starts (at their respective home shards).
This is a contradiction.
Therefore, $T_{i,\gamma} \prec_{L_\gamma} T_{j,\gamma}$, as needed.

\end{proof} 

\begin{theorem}[Liveness]
\label{theorem:liveness}
Our protocol guarantees that every issued transaction will eventually commit.
\end{theorem}

\begin{proof}
Consider the timing assumptions for $\Delta_1$, $\Delta_2$, and $\Delta_3$ as described in Section \ref{sec:preliminaries-and-sharding-model}.
Consider a transaction $T_i$ with ID $ID(T_i)$ generated at time $t$.
In the worst case, $T_i$ will execute when its ID is the lowest in the system,
through force rollback messages.

After $c \cdot \Delta_1$ time steps, 
every new transaction generated will have a larger ID than $ID(T_i)$,
and hence lower precedence than $T_i$.
It takes additional time $\Delta_2$ to propagate $ID(T_i)$.

Let $ID'_{\min}$ be the smallest ID of all transactions considering all the pools of all shards at time $t$.
Let $q$ be the number of transactions which at time $t + \Delta_2 + c \cdot \Delta_1$ have ID at least 
$ID'_{\min}$ and less than $ID(T_i)$.
In the worst case, all of these $q$ transactions may commit before $T_i$.
As we have 7 phases in our protocol, for each committed transaction,
the combined upper bound for communication and consensus delay time is $7(\Delta_1 + \Delta_3)$.
Hence, it takes at most $q \cdot 7(\Delta_1 + \Delta_3)$ time to commit the $q$ transactions.
Therefore, by time $t + q \cdot 7(\Delta_1 + \Delta_3) + \Delta_2 + c \cdot \Delta_1$ 
transaction $T_i$ will be committed in the blockchain.

\end{proof}

\section{Performance Evaluation}
\label{sec:performance-evaluation}
We set up our experiments in a virtual machine in M1 MAC PC with a 10-core CPU and 32-core GPU, including 32 GB RAM. We used Python programming language for the experiments which supports multiprocessing and multithreading. We virtually created multiple shards within a machine and conducted the experiment with different numbers of shards. For the communication between the shards, we use socket programming in Python, which enables the communication between shards by message passing. Same as previous work \cite{Byshard}, we also assume that each shard runs the consensus algorithm and takes 30ms (say $\Delta_3$) for decisions.

We generate 1000 accounts randomly by using the combination of the English alphabet letters and assigned an initial balance of 3000 to each account. Moreover, we generate 1500 transactions by randomly selecting the account from 1000 accounts. Each transaction includes the read and writes operations with some constraints. The generated 1500 transactions are divided with respect to the number of shards and randomly assigned to the transaction pool of each leader shard.

We show the experimental results in three categories. Firstly, optimal (no lock), means there is no transaction isolation; concurrent transactions can access the accounts and update those accounts without any consideration of the data consistency. Secondly, We used the concept of exclusive lock protocol to ensure transaction isolation and concurrency control. This approach acquires a lock on an object (account), at the time of accessing it and releases the lock after the transaction completes \cite{bernstein1981concurrency}. This prevents other transactions from accessing the same object until the lock is released, ensuring that transactions do not interfere with each other. When a transaction acquires an exclusive lock on a data item, no other transaction can read or modify that data item until the lock is released, providing exclusive access to the data. In our implementation, when an object is locked, other transactions attempting to access the same object wait until the lock is released. This guarantees that the transaction holding the lock has exclusive access to the data and can modify it without interference from other transactions. Finally, we used our protocol to achieve transaction isolation and concurrency control without using a lock, which takes a snapshot of each object (account) and if there is conflict occurs then priority to access the object is given to the earliest transaction and other transaction are restart and rollback to re-execute again.

\paragraph{\bf Experimental Results:}
In the first experiment, we evaluate the average throughput of the transactions using  1500 generated transactions, in which each transaction checks whether the account has sufficient balance or not before transferring from its own account to another account, and the other three constraints. If the transaction is valid and satisfies all the conditions then the transaction is executed by removing the balance from one account and adding that balance to another account (i.e. two write operations). To measure the average throughput of transactions, we initialize the start time at the beginning of the transaction processing and capture the final time after processing all the 1500 transactions.

The average throughput of the transaction with respect to the number of shards is shown in Figure \ref{number-of-shards-vs-tps}, where we measure the average throughput of the transactions by varying the number of shards. From the experiment, we observe that the throughput increases with the number of shards. From Figure \ref{number-of-shards-vs-tps} we can see that the transaction throughput of our protocol is better than the lock-based protocol and quite close to the  no-lock protocol.
\begin{figure}[ht]
\begin{minipage}[t]{0.475\columnwidth}
  \includegraphics[width=\linewidth]{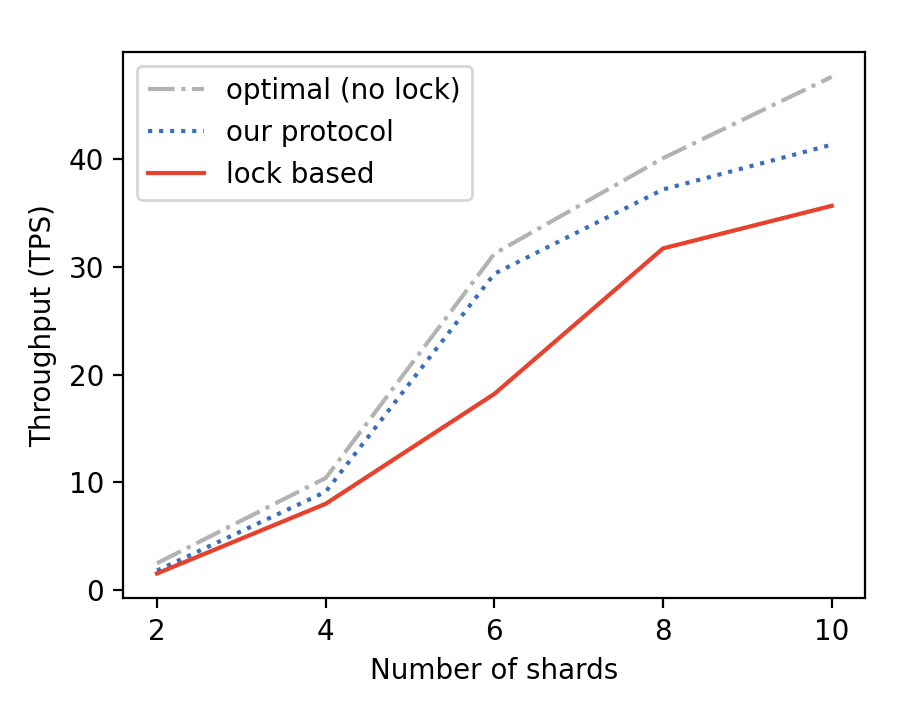}
  \caption{Average transaction throughput with the number of shards}
  \label{number-of-shards-vs-tps}
\end{minipage}\hfill 
\begin{minipage}[t]{0.475\columnwidth}
  \includegraphics[width=\linewidth]{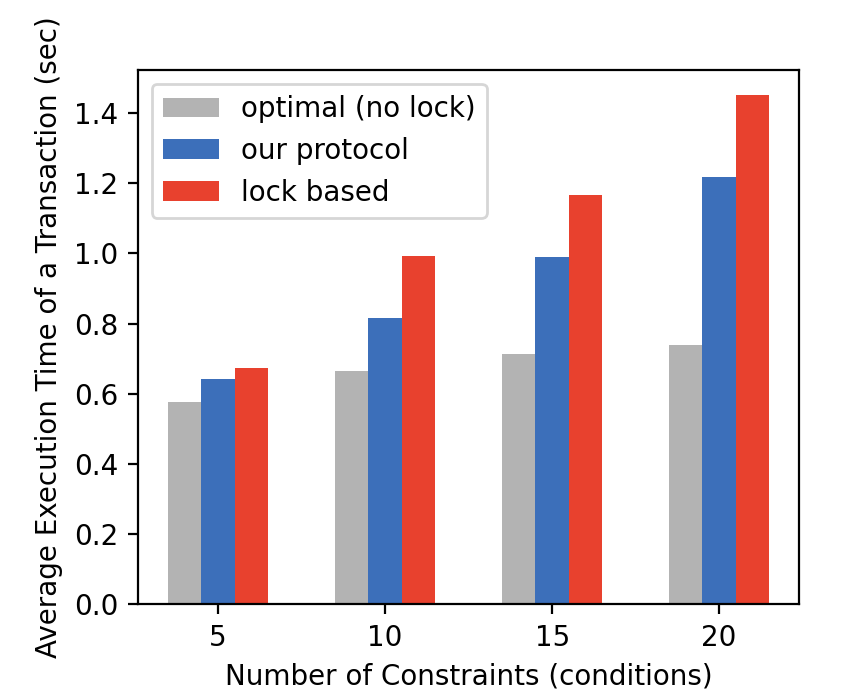}
  \caption{Average execution time of a transaction with numbers of constraints}
  \label{number-of-constraints-vs-execution-time}
\end{minipage}
\end{figure}

In the second experiment, we set up the environment with four shards and calculate the average execution time of a transaction with respect to the number of conditions in each transaction. We increases the constraints of the transactions and recorded the execution of the transactions. In each experiment, we found that the average execution time of our protocol is less than the lock-based protocol also shown  in Figure \ref{number-of-constraints-vs-execution-time}. From the experimental result, we see that as the number of conditions to execute the transaction increased, the commit process takes a long time. As a result, in the lock-based protocol, the lock is kept for a long period, which adds a lot of overhead and takes more time for the execution of transactions than in our proposed protocol.

\section{Conclusion}
\label{sec:conclusion}
In this research work, we presented a lockless transaction scheduling protocol for blockchain sharding.  
Our protocol is based on taking a snapshot version of the various shared objects (accounts) that the transactions access in each shard. We provide a correctness proof with the safety and liveness properties of our protocol.
We also evaluate our protocol experimentally through simulations and we observe that the transaction execution time is considerably faster than the lock-based approaches and also the throughput of the transactions is improved with an increasing number of shards.

This study still has some room for improvement. One possible extension could be a study on efficient communication between leader shards and destination shards.
Introducing a formal performance analysis for blockchain sharding is another interesting topic for future work,
which will quantify the performance based on parameters of the blockchain, such as the number of shards and 
the sizes of the shards.

In recent literature, Schwarzmann \cite{schwarzmann2021towards} reviewed several requirements that need to be satisfied by electronic poll book systems, such as ensuring correctness, security, integrity, fault-tolerance, consistent distributed storage, etc. Our proposed protocol can be used to address some of these issues because it not only provides blockchain features but also offers scalability and better performance for recording transactions. Overall, our protocol may offer many unique features for electronic check-in poll book systems, including decentralization, immutability, and consensus.

\subsection*{Acknowledgements}
This paper is supported by NSF grant CNS-2131538.


%
%
%
%


\bibliographystyle{splncs04}
\bibliography{references}

\begin{thebibliography}{10}
\providecommand{\url}[1]{\texttt{#1}}
\providecommand{\urlprefix}{URL }
\providecommand{\doi}[1]{https://doi.org/#1}

\bibitem{SharPer}
Amiri, M.J., Agrawal, D., El~Abbadi, A.: Sharper: Sharding permissioned
  blockchains over network clusters. In: Proceedings of the 2021 International
  Conference on Management of Data. pp. 76--88 (2021)

\bibitem{distributed-database-sharding}
Bagui, S., Nguyen, L.T.: Database sharding: to provide fault tolerance and
  scalability of big data on the cloud. International Journal of Cloud
  Applications and Computing (IJCAC)  \textbf{5}(2),  36--52 (2015)

\bibitem{sql-isolation-level}
Berenson, H., Bernstein, P., Gray, J., Melton, J., O'Neil, E., O'Neil, P.: A
  critique of ansi sql isolation levels. ACM SIGMOD Record  \textbf{24}(2),
  1--10 (1995)

\bibitem{bernstein1981concurrency}
Bernstein, P.A., Goodman, N.: Concurrency control in distributed database
  systems. ACM Computing Surveys (CSUR)  \textbf{13}(2),  185--221 (1981)

\bibitem{multiversion-concurrency-control}
Bernstein, P.A., Goodman, N.: Multiversion concurrency control—theory and
  algorithms. ACM Transactions on Database Systems (TODS)  \textbf{8}(4),
  465--483 (1983)

\bibitem{serializable-isolation}
Cahill, M.J., R{\"o}hm, U., Fekete, A.D.: Serializable isolation for snapshot
  databases. ACM Transactions on Database Systems (TODS)  \textbf{34}(4),
  1--42 (2009)

\bibitem{PBFT}
Castro, M., Liskov, B., et~al.: Practical byzantine fault tolerance. In: OsDI.
  vol.~99, pp. 173--186 (1999)

\bibitem{two-phase-locking}
Eswaran, K.P., Gray, J.N., Lorie, R.A., Traiger, I.L.: The notions of
  consistency and predicate locks in a database system. Communications of the
  ACM  \textbf{19}(11),  624--633 (1976)

\bibitem{packing-message-as-a-tool}
Friedman, R., Van~Renesse, R.: Packing messages as a tool for boosting the
  performance of total ordering protocols. In: Proceedings. The Sixth IEEE
  International Symposium on High Performance Distributed Computing (Cat. No.
  97TB100183). pp. 233--242. IEEE (1997)

\bibitem{No-Commit-Proofs}
Giridharan, N., Howard, H., Abraham, I., Crooks, N., Tomescu, A.: No-commit
  proofs: Defeating livelock in bft. Cryptology ePrint Archive  (2021)

\bibitem{transaction-processing}
Gray, J., Reuter, A.: Transaction Processing: Concepts and Techniques. Morgan
  Kaufmann Publishers Inc., San Francisco, CA, USA, 1st edn. (1992)

\bibitem{SBFT}
Gueta, G.G., Abraham, I., Grossman, S., Malkhi, D., Pinkas, B., Reiter, M.,
  Seredinschi, D.A., Tamir, O., Tomescu, A.: Sbft: a scalable and decentralized
  trust infrastructure. In: 2019 49th Annual IEEE/IFIP international conference
  on dependable systems and networks (DSN). pp. 568--580. IEEE (2019)

\bibitem{Cerberus}
Hellings, J., Hughes, D.P., Primero, J., Sadoghi, M.: Cerberus: Minimalistic
  multi-shard byzantine-resilient transaction processing. arXiv preprint
  arXiv:2008.04450  (2020)

\bibitem{Byshard}
Hellings, J., Sadoghi, M.: Byshard: Sharding in a byzantine environment.
  Proceedings of the VLDB Endowment  \textbf{14}(11),  2230--2243 (2021)

\bibitem{Jalal-Window}
{Jalalzai}, M.M., {Busch}, C.: Window based {BFT} blockchain consensus. In:
  iThings, IEEE GreenCom, IEEE (CPSCom) and IEEE SSmartData 2018. pp. 971--979
  (July 2018)

\bibitem{jalalzai2019proteus}
Jalalzai, M.M., Busch, C., Richard, G.G.: Proteus: A scalable bft consensus
  protocol for blockchains. In: 2019 IEEE international conference on
  Blockchain (Blockchain). pp. 308--313. IEEE (2019)

\bibitem{jalalzai2021hermes}
Jalalzai, M.M., Feng, C., Busch, C., Richard, G.G., Niu, J.: The hermes bft for
  blockchains. IEEE Transactions on Dependable and Secure Computing
  \textbf{19}(6),  3971--3986 (2021)

\bibitem{OmniLedger}
Kokoris-Kogias, E., Jovanovic, P., Gasser, L., Gailly, N., Syta, E., Ford, B.:
  Omniledger: A secure, scale-out, decentralized ledger via sharding. In: 2018
  IEEE Symposium on Security and Privacy (SP). pp. 583--598. IEEE (2018)

\bibitem{optimistic-menhods-for-concurrency-control}
Kung, H.T., Robinson, J.T.: On optimistic methods for concurrency control. ACM
  Transactions on Database Systems (TODS)  \textbf{6}(2),  213--226 (1981)

\bibitem{Elastico}
Luu, L., Narayanan, V., Zheng, C., Baweja, K., Gilbert, S., Saxena, P.: A
  secure sharding protocol for open blockchains. In: Proceedings of the 2016
  ACM SIGSAC conference on computer and communications security. pp. 17--30
  (2016)

\bibitem{lockless-transaction-hyperledger-fabric}
Meir, H., Barger, A., Manevich, Y., Tock, Y.: Lockless transaction isolation in
  hyperledger fabric. In: 2019 IEEE International Conference on Blockchain
  (Blockchain). pp. 59--66 (2019). \doi{10.1109/Blockchain.2019.00017}

\bibitem{bitcoin}
Nakamoto, S.: Bitcoin : A peer-to-peer electronic cash system (2009)

\bibitem{blockchain-hype-vs-reality}
Pisa, M., Juden, M.: Blockchain and economic development: Hype vs. reality.
  Center for Global Development Policy Paper  \textbf{107}, ~150 (2017)

\bibitem{serializable-snapshot-isolation}
Ports, D.R., Grittner, K.: Serializable snapshot isolation in postgresql. arXiv
  preprint arXiv:1208.4179  (2012)

\bibitem{survey-of-onsensus}
Sankar, L.S., Sindhu, M., Sethumadhavan, M.: Survey of consensus protocols on
  blockchain applications. In: 2017 4th international conference on advanced
  computing and communication systems (ICACCS). pp.~1--5. IEEE (2017)

\bibitem{schwarzmann2021towards}
Schwarzmann, A.A.: Towards a robust distributed framework for election-day
  voter check-in. In: Stabilization, Safety, and Security of Distributed
  Systems: 23rd International Symposium, SSS 2021, Virtual Event, November
  17--20, 2021, Proceedings 23. pp. 173--193. Springer (2021)

\bibitem{Rapidchain}
Zamani, M., Movahedi, M., Raykova, M.: Rapidchain: Scaling blockchain via full
  sharding. In: Proceedings of the 2018 ACM SIGSAC conference on computer and
  communications security. pp. 931--948 (2018)

\end{thebibliography}

\end{document}